\documentclass[accepted]{uai2025} 
                        
\usepackage[american]{babel}

\usepackage{natbib} 
\usepackage{mathtools} 
\usepackage{booktabs} 
\usepackage{tikz} 
\usepackage{xspace}

\usepackage{titlesec}
\usepackage{titletoc}

\usepackage[utf8]{inputenc} 
\usepackage[T1]{fontenc}    
\usepackage{hyperref}       
\usepackage{url}            
\usepackage{booktabs}       
\usepackage{amsfonts}       
\usepackage{nicefrac}       
\usepackage{microtype}      
\usepackage{xcolor}         
\usepackage{amsmath}
\usepackage{graphicx}
\usepackage{cleveref}
\usepackage{algorithm}
\usepackage{amsthm}
\newtheorem{theorem}{Theorem}[section]  
\newtheorem{assumption}{Assumption}
\newtheorem{lemma}{Lemma}

\usepackage{thmtools}
\usepackage{thm-restate}
\usepackage[T1]{fontenc}

\newtheorem{definition}{Definition}[section] %
\newtheorem{proposition}{Proposition}[section] 
\usepackage{amsthm}

\usepackage{amsmath,amssymb}
\usepackage{algorithm}
\usepackage{algpseudocode}  

\newcommand{\ours}{\textsc{DiffOracle}\xspace}

\title{Robust Optimization with Diffusion Models for Green Security}

\author[1]{\href{lingkaikong@g.harvard.edu}{Lingkai Kong}{}}
\author[1]{Haichuan Wang}
\author[1]{Yuqi Pan}
\author[1]{Cheol Woo Kim}
\author[1]{Mingxiao Song}
\author[1]{Alayna Nguyen}
\author[1]{\\Tonghan Wang}
\author[2]{Haifeng Xu}
\author[1]{Milind Tambe}
\affil[1]{%
    Harvard University
}
\affil[2]{%
    University of Chicago

}

\begin{document}
\maketitle

\begin{abstract}

In green security, defenders must forecast adversarial behavior—such as poaching, illegal logging, and illegal fishing—to plan effective patrols. These behaviors are often highly uncertain and complex. Prior work has leveraged game theory to design robust patrol strategies to handle uncertainty, but existing adversarial behavior models primarily rely on Gaussian processes or linear models, which lack the expressiveness needed to capture intricate behavioral patterns.  To address this limitation, we propose a conditional diffusion model for adversary behavior modeling, leveraging its strong distribution-fitting capabilities. To the best of our knowledge, this is the first application of diffusion models in the green security domain.  
Integrating diffusion models into game-theoretic optimization, however, presents new challenges, including a constrained mixed strategy space and the need to sample from an unnormalized distribution to estimate utilities. To tackle these challenges, we introduce a mixed strategy of mixed strategies and employ a twisted Sequential Monte Carlo (SMC) sampler for accurate sampling.  Theoretically, our algorithm is guaranteed to converge to an \(\epsilon\)-equilibrium with high probability using a finite number of iterations and samples. Empirically, we evaluate our approach on both synthetic and real-world poaching datasets, demonstrating its effectiveness.

\end{abstract}

\section{Introduction}

In green security, mitigating threats such as illegal logging, illegal fishing, poaching, and environmental pollution requires defenders to anticipate and counteract adversarial behaviors \citep{IJCAI15-fei}. For example, in wildlife conservation, rangers must predict poachers' movements and then strategically allocate patrols to protect endangered species. Over the years, numerous predictive models have been developed \citep{kar2017cloudy, gurumurthy2018exploiting, xu2020stay}, alongside robust patrol optimization methods that leverage game theory to enhance decision-making based on these predictions \citep{pmlr-v161-xu21a}.


However, existing adversary predictive models~\citep{kar2017cloudy, gurumurthy2018exploiting}  either lack uncertainty quantification~\citep{kong2023uncertainty, li2023muben, kong2024two} or provide only parameterized predictive distributions with limited expressiveness \citep{xu2020stay}. In reality, adversarial behaviors—such as those of poachers—are high-dimensional and highly complex, driven by diverse motivations, constraints, and strategies. Capturing the full extent of uncertainty is particularly challenging in the strategic environments, where conventional models may struggle to account for the variability of real-world threats.

In this work, we propose using diffusion models to capture adversarial behavior in green security. Diffusion models are a powerful framework for modeling complex, high-dimensional, non-parametric distributions, and they have been successfully applied to image modeling~\citep{ho2020denoising, rombach2021high}, video generation~\citep{ho2022video}, and time-series forecasting~\citep{yang2024survey}. By iteratively refining samples through a denoising process, diffusion models can generate diverse and plausible scenarios, offering a more comprehensive representation of potential attacker strategies. To the best of our knowledge, ours is the first attempt to apply diffusion models in green security.

To enhance the robustness of our approach against potential errors in the learned diffusion model (arising from noisy data, limited sample sizes, or imperfect network training), we assume the attacker's true mixed strategy lies within a KL-divergence ball centered around the learned model distribution. We then optimize for the worst-case expected utility within this constrained space. This formulation naturally gives rise to a two-player zero-sum game: while a defender aims to maximize the expected utility, a nature adversary selects the mixed strategy from the KL ball to minimize it.

This game-theoretic formulation involving diffusion models introduces new technical challenges that have not been addressed in the literature.
First, the KL-divergence constraint on the adversary's mixed strategy prevents the direct application of the standard double oracle method. To resolve this, we shift the constraint from the mixed strategy space to the pure strategy space, treating the original mixed strategy as a pure strategy and introducing a ``mixed strategy over mixed strategies.'' This reformulation yields a more tractable optimization problem. Another challenge arises from the need to sample from a reweighted version of the diffusion model to estimate utilities. To address this, we employ twisted sequential Monte Carlo (SMC) sampling, ensuring asymptotic correctness when evaluating the relevant expected utilities.



Our contributions are as follows:  
(1) \textbf{Novel Adversary Modeling:} We are the first to leverage diffusion models for modeling adversarial behavior in green security domains.  
(2) \textbf{Robust Optimization with Diffusion Model Framework:} We propose \ours to mitigate imperfections in learned adversary models by introducing a double oracle algorithm that efficiently computes robust mixed patrol strategies. 
(3) \textbf{Theoretical Guarantees:} We prove that our method converges to an \(\epsilon\)-equilibrium with high probability under a finite number of iterations and a finite number of samples. (4) \textbf{Empirical Performance:} We empirically evaluate our method on both synthetic and real-world poaching data.
\section{Related Works}


\textbf{Diffusion Models}
Diffusion models have achieved remarkable success across various generative modeling tasks, including image generation~\citep{song2021scorebased, ho2020denoising}, decision-making~\citep{kong2024diffusion, kong2025composite}, and scientific discovery~\citep{gruver2024protein, watson2023novo, kong2023autoregressive}. These models are particularly adept at capturing complex, high-dimensional distributions, making them a powerful tool for diverse applications. Conditional diffusion models extend this capability by integrating contextual information to guide the generative process. By conditioning on textual descriptions, semantic masks, or other relevant features, these models enable tasks such as text-to-image generation~\citep{saharia2022photorealistic}, image-to-image translation~\citep{saharia2021palette}, and time series forecasting~\citep{shen2023non}.


\textbf{Double Oracle for Robust Optimization}
Prior work has framed robust optimization as a two-player zero-sum game~\citep{mastin2015randomized, gilbert2017double}, where the optimizing player selects a potentially randomized feasible strategy, while an adversary chooses problem parameters to maximize regret. The double oracle (DO) algorithm is a standard method for computing equilibria in such games ~\citep{mcmahan2003planning,adam2021double} and has been applied to robust influence maximization in social networks~\citep{wilder2017uncharted}, robust patrol planning~\citep{pmlr-v161-xu21a}, robust submodular optimization~\citep{wilder2018equilibrium}, and robust policy design for restless bandits~\citep{killian2022restless}. However, these applications restrict the uncertainty set to a compact interval. In contrast, our problem involves a diffusion model that provides full distribution-level predictions, making the uncertainty set a space of distributions, which introduces new theoretical challenges in applying double oracle.

\textbf{Distributionally Robust Optimization}
Our work is also closely related to Distributionally Robust Optimization (DRO) \citep{rahimian2019distributionally}, which seeks to find robust solutions by optimizing for the worst-case scenario over a set of plausible distributions, known as the ambiguity set. This framework is particularly effective for handling uncertainty and distributional shifts in optimization objectives or constraints. DRO has seen widespread application in areas such as supply chain management \citep{ash2022distributionally}, finance \citep{kobayashi2023cardinality}, and machine learning \citep{madry2018towards, Sagawa*2020Distributionally}, where resilience to data perturbations is critical. However, most existing DRO methods focus on identifying a single pure strategy, which  is dangerous in the green security setting that adversaries can learn to anticipate and exploit. To address this, we propose a game-theoretic approach that derives a mixed strategy for the defender, leveraging randomness to enhance unpredictability and bolster robustness against adversarial exploitation.


\textbf{Green Security Games}
Green Security Games (GSGs) use game-theoretic frameworks to safeguard valuable environmental resources from illegal activities such as poaching and illegal fishing \citep{IJCAI15-fei, hasan2022evaluation}. In these settings, a resource-limited defender protects expansive, spatially distributed areas against attackers with bounded rationality. Prior work focused on forecasting poaching behaviors \citep{gurumurthy2018exploiting, moore2018ranger}, learning attacker behavior models from data \citep{nguyen2016capture, gholami2018adversary, xu2020stay}, designing patrol strategies \citep{IJCAI15-fei, GameSec17-haifeng}, and balancing data collection with poaching detection \citep{Xu2020DualMandatePM}.

Among existing studies, \citet{pmlr-v161-xu21a} is most closely related to ours, as it also employs a double oracle method to design robust patrolling strategies. However, our approach differs in two key ways. First, we are the first to use diffusion models to predict poaching behavior, addressing the limited expressiveness of the linear approach in \citet{pmlr-v161-xu21a}. Second, while \citet{pmlr-v161-xu21a} focuses on minimax regret with interval-shaped uncertainty sets, our work adopts a distributionally robust optimization objective.





\section{Preliminaries on Diffusion Model}
\label{sec:diffusion}




A diffusion model \citep{sohl2015deep} is a generative framework composed of two stochastic processes: a \emph{forward} process that progressively adds Gaussian noise to real data, and a \emph{reverse} (or denoising) process that learns to remove this noise step by step. Formally, let \(\mathbf{z}^0 \sim \mathcal{D}\) be a sample from the training dataset.\footnote{We use \(\mathbf{z}^0\) and \(\mathbf{z}\) interchangeably when there is no ambiguity.} The forward diffusion process can be written as
$q(\mathbf{z}^t \mid \mathbf{z}^{t-1}) 
= \mathcal{N}\!\bigl(\mathbf{z}^t;\,\mathbf{z}^{t-1},\,\beta^2 \mathbf{I}\bigr),$
where \(\beta^2\) is the noise variance at each step \(t=1,\dots,T\). As \(T\) becomes large, repeated noising transforms the data distribution into (approximately) pure Gaussian noise:
$q(\mathbf{z}^T) \approx \mathcal{N}(\mathbf{0},\,T \beta^2 \mathbf{I}).$

\textbf{Score-based Approximation.} To invert this process (i.e., to denoise and recover samples from the original data distribution), one can approximate the reverse transition 
\(
q(\mathbf{z}^{t-1} \mid \mathbf{z}^t) 
\)
via the \emph{score function}, \(\nabla_{\mathbf{z}^t} \log q(\mathbf{z}^t)\) when $\beta$ is small. Specifically,
\[
q(\mathbf{z}^{t-1} \mid \mathbf{z}^t) 
\,\approx\, 
\mathcal{N}\!\Bigl(\mathbf{z}^{t-1};\,
\mathbf{z}^t 
+ \beta^2 \,\nabla_{\mathbf{z}^t}\! \log q(\mathbf{z}^t),\,
\beta^2 \mathbf{I}\Bigr).
\]
Here, \(q(\mathbf{z}^t) = \int q(\mathbf{z}^0)\,q(\mathbf{z}^t \mid \mathbf{z}^0)\, d\mathbf{z}^0\), and the gradient \(\nabla_{\mathbf{z}^t}\! \log q(\mathbf{z}^t)\) points toward regions of higher data density. In practice, we do not know \(q(\mathbf{z}^t)\) in closed form, so a neural \emph{score network} \(s_{\theta}(\mathbf{z}^t, t)\) is trained to approximate this gradient via \emph{denoising score matching} \citep{vincent2011connection, ho2020denoising}. Consequently, the learned reverse (denoising) transition becomes
\[
p_{\theta}(\mathbf{z}^{t-1} \mid \mathbf{z}^t) 
= \mathcal{N}\!\Bigl(\mathbf{z}^{t-1};\,
\mathbf{z}^t 
+ \beta^2\, s_{\theta}(\mathbf{z}^t, t),\,
\beta^2 \mathbf{I}\Bigr).
\]
Starting from an initial Gaussian sample \(\mathbf{z}^T \sim \mathcal{N}(\mathbf{0},\, T \beta^2 \mathbf{I})\), iterating this reverse process ultimately recovers samples that approximate the original data distribution.

\textbf{Conditional Extension.} This diffusion framework can be naturally extended to include additional context \(\mathbf{c}\). In a \emph{conditional} diffusion model~\citep{ho2021classifier}, the score network becomes \(s_{\theta}(\mathbf{z}^t, t, \mathbf{c})\), so that at each step the denoising is informed by side information such as class labels, textual descriptions, or other relevant features. This conditional approach enables the generation of samples that match not only the learned data distribution but also the specific context \(\mathbf{c}\), making it particularly useful for tasks in which external conditions strongly influence the underlying data generation process.

\section{Problem Formulation}
In green security settings, a defender (e.g., a ranger) patrols a protected area to prevent resource extraction by an attacker (e.g., a poacher or illegal logger). Let \(K\) denote the number of targets—such as \(1 \times 1\) km regions within the protected area—that require protection. The defender must allocate patrol effort across these targets while adhering to a total budget \(B\). Formally, the patrol strategy is represented as \(\mathbf{x} = (x_1, \dots, x_K)\), where \(x_k\) denotes the amount of effort (e.g., patrol hours) assigned to target \(k\). The defender’s strategy is constrained by: $\mathcal{X} = \{\mathbf{x} \in \mathbb{R}^K \mid x_k \geq 0, \forall k, \sum_{k=1}^{K} x_k \leq B\}$, which ensures that all patrol efforts are non-negative and do not exceed the available budget $B$.

\textbf{Attacker Behavior via a Conditional Diffusion Model.} Building on the diffusion-model framework in Section~\ref{sec:diffusion}, we now focus on a poaching scenario, in which an attacker’s behavior can be highly uncertain and multimodal. Let \(\mathbf{z}\) denote the number of snares or traps placed in each \(1 \times 1\) km region, where \(K\) is its dimensionality. Similarly, let \(\mathbf{c}\) represent contextual features, including last month’s patrol effort per region~\citep{pmlr-v161-xu21a,xu2020stay}, distance to the park boundary, elevation, and land cover. We model the attacker’s behavior with a continuous \emph{conditional diffusion model} \(p_{\theta}(\mathbf{z} \mid \mathbf{c})\). Concretely, we treat historical poaching data as samples of \(\mathbf{z}^0\), add noise in a forward process, and learn a reverse (denoising) process conditioned on \(\mathbf{c}\). Once trained, this diffusion model captures how poachers respond to different patrol allocations and environmental factors. By sampling from \(p_{\theta}(\mathbf{z} \mid \mathbf{c})\) for new contexts, we can generate realistic, diverse poaching scenarios that inform patrol strategy design. Table~\ref{table:mse} shows the forecasting results on the real-world poaching data and we can see the \textbf{diffusion model can outperform existing approaches used in green security}. Experimental details are in Appendix.~\ref{sec:exp-details}.

\begin{table}[h]
\footnotesize
\centering
\begin{tabular}{@{}ccc@{}}
\toprule
\textbf{Model} & \textbf{MSE} \\ \midrule
Linear regression  & $24.40$ \\
Gaussian process   & $24.21\pm 0.04$ \\
Diffusion model    & $\mathbf{23.46} \pm 0.07$ \\ \bottomrule
\end{tabular}
\caption{Forecasting performance in terms of mean squared error (MSE) on poaching data.}
\label{table:mse}
\end{table}

\textbf{Robust Patrol Optimization.} In practice, the learned diffusion model may be imperfect due to data noise, limited training samples, or suboptimal network training. Consequently, the learned distribution might not accurately capture the true underlying behavior. To ensure robustness in patrol strategy design, we assume the true distribution lies within a bounded KL-divergence from the learned distribution. We then optimize for the worst-case expected utility over all distributions in that KL-divergence ball, leading to the following formulation:
\begin{align}
 \textstyle &\max_{\pi(\mathbf{x})\in \Delta(\mathcal{X})} \min_{\tau(\mathbf{z})\in \mathcal{T}} 
\quad  \mathbb{E}_{\pi(\mathbf{x})} \mathbb{E}_{\tau(\mathbf{z})} [u(\mathbf{x}, \mathbf{z})]  \nonumber\\
&\mathcal{T} = \left\{ \tau(\mathbf{z}) \mid D_{\rm KL}(\tau(\mathbf{z}) \parallel p_{\theta}(\mathbf{z} \mid \mathbf{c})) \leq \rho \right\},
\label{eq:DRO}
\end{align}
where \(\rho\) is a tolerance parameter specifying how far the true distribution may deviate from the learned distribution.  \(u(\mathbf{x}, \mathbf{z})\) represents the defender’s utility (e.g., the number of animals in the park) for strategy \(\mathbf{x}\) given the adversary’s choice \(\mathbf{z}\) and is assumed to be bounded in \([0, M]\).

Eq.~\eqref{eq:DRO} can be interpreted as a two-player zero-sum game in which the \emph{defender} seeks a robust mixed strategy \(\pi(\mathbf{x})\), while a  \emph{nature adversary} (representing model misspecification) selects \(\tau(\mathbf{z})\) within the KL-divergence ball to minimize the defender’s expected utility. The defender’s pure strategy space is \(\mathcal{X}\), and the adversary’s pure strategy space is \(\mathcal{Z} = \mathrm{Support}(p_{\theta}(\mathbf{z}|\mathbf{c}))\). The defender’s mixed strategy \(\pi(\mathbf{x})\) is a probability distribution over \(\mathcal{X}\), with the corresponding space denoted by \(\Delta(\mathcal{X})\). In contrast, the adversary’s mixed strategy \(\tau(\mathbf{z})\) is in the constained space $\mathcal{T}$.

For notational simplicity, when both players use mixed strategies, the defender’s expected utility is denoted by \(U(\pi, \tau)\). If one player employs a pure strategy and the other a mixed strategy, we write \(U(\mathbf{x}, \tau) \coloneq U(\delta_{\mathbf{x}}, \tau)\).

Note that, unlike standard DRO, where the goal is typically to find a single strategy $\mathbf{x}$, here we aim to identify a \emph{mixed} strategy for the defender. This approach is particularly important in green security settings, as adopting a randomized policy helps prevent predictability. A deterministic patrol strategy could be exploited by adversaries, such as poachers, who can adapt their behavior to bypass predictable patterns. By introducing randomness into the patrol strategy, we increase the difficulty for adversaries to anticipate the ranger's actions, thereby enhancing the overall security and effectiveness of the patrol.

\begin{figure*}[h]
    \centering
\includegraphics[width=0.95\textwidth]{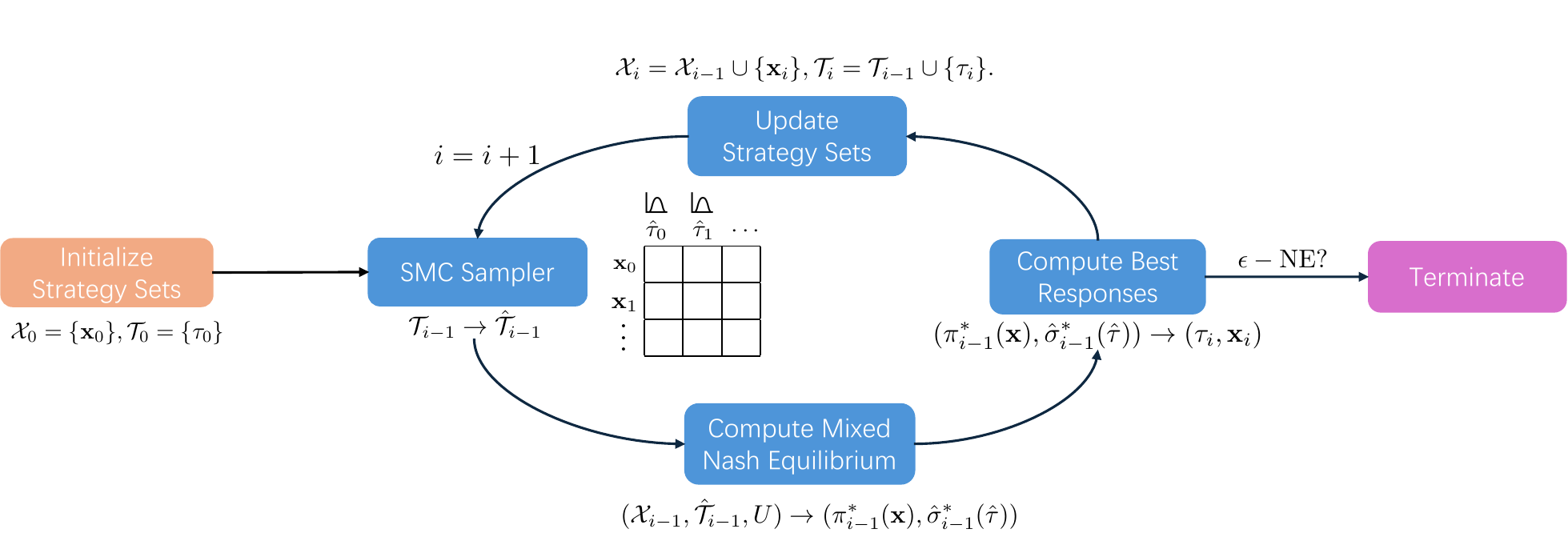} 
    \caption{Overview of \ours. We begin by initializing the strategy set for each player. At the \(i\)-th iteration, we use SMC sampler to obtain a set of empirical distributions \(\hat{\mathcal{T}}_{i-1}\). Next, a mixed Nash solver computes the equilibrium \(\pi^*_{i-1}\) and \(\hat{\sigma}^*_{i-1}\). We then compute each player’s best response against the opponent’s mixed strategy and update the players’ strategy sets. This procedure is repeated until convergence.}
    \label{fig:overview}
\end{figure*}


\section{Robust Optimization with Diffusion Model}

In this section, we propose \ours to solve the robust optimization problem in Eq.~\ref{eq:DRO}. In Section \ref{sec:mixed_over_mixed}, we introduce a mixed strategy over mixed strategies to ensure the applicability of the double oracle approach. Section \ref{sec:double_oracle} details the overall workflow of the algorithm. In Section \ref{thm:twisted_diffusion_sampler}, we present twisted SMC sampler to estimate the expected utility. Finally, in Section \ref{sec:convergence_analysis}, we provide a convergence analysis of \ours.

\subsection{Mixed strategy over mixed strategies}\label{sec:mixed_over_mixed}
Eq.~\eqref{eq:DRO} requires solving for mixed strategies in a continuous game with infinitely many strategies. A common approach for such problems is the double oracle method~\citep{adam2021double}, which iteratively expands both players’ strategy sets and computes the equilibrium of the resulting subgame. This procedure is guaranteed to converge to an equilibrium in any two-player zero-sum continuous game. However, during the double oracle process, the mixed strategy it produces is necessarily a discrete distribution, whereas $p_{\theta}(\mathbf{z}|\mathbf{c})$ is a continuous distribution. As a result, the KL divergence between these two distributions is ill-defined, making it difficult to include the KL-divergence constraint in the subgame-equilibrium computation.

To address this limitation, we note that given a fixed \(\pi(\mathbf{z})\), the inner constrained minimization problem admits a closed-form solution that can be sampled using the diffusion model. This procedure can be interpreted as computing a best-response pure strategy in the double oracle framework. Consequently, we propose viewing the original mixed strategy \(\tau(\mathbf{z})\) as a “pure” strategy and introducing a \emph{mixed strategy over mixed strategies}. This reformulation enables the application of the double oracle method while preserving the desired constraints.






\begin{definition}[Mixed Strategy over Mixed Strategies]\label{def:mixed_over_mixed}
Let \(\mathcal{T} \) denote the space of mixed strategies, where each \( \tau \in \mathcal{T} \) represents a probability distribution over pure strategies. A \textit{mixed strategy over mixed strategies}, \( \sigma \), is a probability distribution over \( \mathcal{T}\), formally expressed as \( \sigma \in \Delta(\mathcal{T}) \). This implies that \( \sigma \) satisfies the following conditions: (1) \( \sigma(\tau) \geq 0 \) for all \( \tau \in \Delta \), and (2) \( \int_{\mathcal{T}} \sigma(\tau) \, d\tau = 1 \).  
\end{definition}

We provide concrete examples in Appendix \ref{appdx:example} to help readers understand 
Definition \ref{def:mixed_over_mixed}. By introducing this concept of a mixed strategy over mixed strategies, \( \sigma \), we can reformulate our objective as follows:  
\begin{align}
\max_{\pi(\mathbf{x})\in\Delta(\mathcal{X})} \min_{\sigma(\tau)\in \Delta({\mathcal{T})}}  
\mathbb{E}_{\pi(\mathbf{x})} \mathbb{E}_{\sigma(\tau)} \left( \mathbb{E}_{\tau(\mathbf{z})}
\left[ u(\mathbf{x}, \mathbf{z}) \right] \right)\notag \\
\mathcal{T} \;=\; \{\, \tau(\mathbf{z}) \;|\; D_{\mathrm{KL}}(\tau(\mathbf{z}) \,\|\, p_{\theta}(\mathbf{z}\mid \mathbf{c})) \leq \rho \}, 
\label{eq:new-formulation}
\end{align}
In this reformulation, the adversary’s pure strategy is no longer a single value but instead a full distribution \(\tau(\mathbf{z})\). Consequently, the adversary’s pure strategy space becomes $\mathcal{T}$
and the corresponding mixed strategy space is the set of distributions over these distributions, \(\Delta(\mathcal{T})\). Under this framework, the defender’s utility function takes the form \(\mathbb{E}_{\tau(\mathbf{z})}\bigl[u(\mathbf{x}, \mathbf{z})\bigr]\), while the attacker’s utility becomes \(-\mathbb{E}_{\tau(\mathbf{z})}\bigl[u(\mathbf{x}, \mathbf{z})\bigr]\).

Crucially, this reformulation shifts the KL-divergence constraint from the adversary’s mixed strategy space to its pure strategy space. As we will show in Section~\ref{sec:double_oracle}, the best response for such a constrained pure strategy can be written in closed form. Hence, Eq.~\eqref{eq:new-formulation} can be solved efficiently using the double oracle algorithm.

\begin{proposition}\label{thm:mixed_over_mixed}
The reformulated objective in Eq.~\eqref{eq:new-formulation} yields the same defender mixed strategy \(\pi(\mathbf{x})\) as the original formulation in Eq.~\eqref{eq:DRO}.

Proof. See Appendix.~\ref{appdx:mixed_over_mixed}.
\label{theorem:mixed}
\end{proposition}

By Proposition~\ref{theorem:mixed}, solving Eq.~\eqref{eq:new-formulation} is equivalent to solving Eq.~\eqref{eq:DRO}. Therefore, applying the double oracle algorithm to Eq.~\eqref{eq:new-formulation} recovers the optimal defender mixed strategy for the original problem (Eq.~\eqref{eq:DRO}).

Since we have reformulated the problem, we will henceforth refer to the adversary's pure strategy as \(\tau(\mathbf{z})\) and the mixed strategy as \(\sigma(\tau)\).

\subsection{Double Oracle Flow}\label{sec:double_oracle}


The overall double oracle algorithm is outlined in Algorithm~\ref{alg:double-oracle} and illustrated in Figure~\ref{fig:overview}. We begin by initializing the adversary’s strategy as \(\tau_0 = p_{\theta}(\mathbf{z}|\mathbf{c}) \) and selecting a random initial defender strategy \(\mathbf{x}_0\) from \(\mathcal{X}\) (lines 2-3), forming the initial strategy sets \(\mathcal{T}_0\) and \(\mathcal{X}_0\). These serve as the foundation for the iterative process.
In each iteration, we first sample from each distribution in \(\mathcal{T}_{i-1} = \{\tau_0, \dots, \tau_{i-1}\}\) to obtain a set of empirical distributions \(\hat{\mathcal{T}}_{i-1} = \{\hat{\tau}_0, \dots, \hat{\tau}_{i-1}\}\) (line 6). These empirical distributions are used to estimate expected utilities, which are then input into the \textit{Mixed Nash Equilibrium solver} to compute an equilibrium \((\pi_{i-1}^*, \hat{\sigma}^*_{i-1})\) of the subgame \(\{\mathcal{X}_{i-1}, \hat{\mathcal{T}}_{i-1}, U \}\) (lines 7). Next, the \textit{defender oracle} and \textit{attacker oracle} compute their respective best responses to the mixed strategy, yielding new strategies \(\mathbf{x}_i\) and \(\tau_i(\mathbf{z})\) (lines 8-9). These best response strategies are then added to the strategy sets, expanding them to \(\mathcal{X}_i\) and \(\mathcal{T}_i\) (line 10). 

This iterative procedure alternates between the oracles and the solver until convergence (lines 11–13). The parameters \(\mathsf{prob}\) and tolerance $\epsilon$ are user-defined and guarantee that the algorithm converges to an \(4\epsilon\)-equilibrium with probability \(1 - \mathsf{prob}\), as detailed in Theorem~\ref{thm:convergence_inf_round}. In practice, to manage runtime, we cap the number of double oracle iterations to a fixed limit—a common strategy also employed in \citet{lanctot2017unified, pmlr-v161-xu21a}. 

We introduce the details of the three key components defender oracle, adversary oracle and Mixed Nash equilibrium solver as below.

\textbf{Adversary Oracle}  
At the \(i\)-th iteration, given the defender's mixed strategy $\pi_{i-1}^*$, 
the adversary oracle computes the best response by solving:
\begin{align}
&\tau_{i}(\mathbf{z}) 
= \arg\max_{\tau\in \mathcal{T}} U(\pi^{*}_{i-1}, \tau) \notag \\
&\mathcal{T} = \left\{ \tau_i(\mathbf{z}) \mid D_{\rm KL}(\tau_i(\mathbf{z}) \parallel p_{\theta}(\mathbf{z} \mid \mathbf{c})) \leq \rho \right\},
\label{eq:nature}
\end{align}

\begin{proposition}
\label{propo:kl}
The optimal solution $\tau_i(\mathbf{z})$ of \ref{eq:nature} has a closed-form:
\begin{align}
\textstyle\tau_i(\mathbf{z}) \propto p_{\theta}(\mathbf{z} | \mathbf{c}) \exp\left(-\gamma U(\pi^*_{i-1}, \mathbf{z})\right),
\label{eq:bo-poacher}
\end{align}
where \(\gamma\) is the Lagrange multiplier associated with the KL-divergence constraint.

Proof. See Appendix~\ref{appdx:propo_kl}
\end{proposition}

As shown in Eq.~\ref{eq:bo-poacher}, \(\tau_i(\mathbf{z})\) is an unnormalized distribution obtained by reweighting the original diffusion model distribution according to the utility function. Computing expected utilities under \(\tau_i(\mathbf{z})\) requires sampling from this high-dimensional unnormalized distribution, which is challenging in practice. To address this, we employ twisted Sequential Monte Carlo (SMC) techniques~\citep{chopin2020introduction,wu2023practical}, detailed in Section 4.2, which provide asymptotically exact utility estimates. We denote the resulting empirical distribution as \(\hat{\tau}_i(\mathbf{z}).\)

\textbf{Defender Oracle} 
At the \(i\)-th epoch, given the attacker's mixed strategy \(\hat{\sigma}^*_{i-1}\), the defender oracle computes the best response by solving:  
\begin{align}
\label{eq:defender_oracle}
\vspace{-1em}
\mathbf{x}_{i} = \arg\max_{\mathbf{x}\in\mathcal{X}} U(\mathbf{x}, \hat{\sigma}^*_{i-1}).
\vspace{-1em}
\end{align}
Since \(\hat{\sigma}^*_{i-1}\) represents a mixed strategy over a set of empirical distributions $\hat{\mathcal{T}}_{i-1}$, we can directly compute the expected utility, reducing the problem to a standard deterministic optimization. To  handle the budget constraint in our setting, we employ mirror ascent~\citep{nemirovski2012tutorial}.

\textbf{Mixed Nash Equilibrium Solver}  
At $i$-th iteration, the Mixed Nash Equilibrium solver computes a mixed Nash equilibrium  $(\pi^*_{i-1}, \hat{\sigma}^*_{i-1})$ over the players' current strategy sets $\mathcal{X}_{i-1}$ and $\hat{\mathcal{T}}_{i-1}$. 
The equilibrium can be found using linear programming \citep{nisan2007algorithmic}, and in our work, we utilize the PuLP implementation~\citep{pulp} for this purpose. 
\begin{algorithm}[!t]
\small
\caption{Double Oracle with Diffusion Models}
\label{alg:double-oracle}
\begin{algorithmic}[1]
\Require Pretrained diffusion model $p_{\theta}(\mathbf{z} \mid \mathbf{c})$, utility function $U(\mathbf{x}, \tau)$, probability threshold $\mathsf{prob} > 0$ 
\State Initialize $i \gets 0$
\State $\mathbf{x}_0 \gets \text{random strategy}$, \quad $\tau_0 \gets p_{\theta}(\mathbf{z}\mid \mathbf{c})$
\State $\mathcal{X}_0 \gets \{\mathbf{x}_0\}$, \quad $\mathcal{T}_0 \gets \{\tau_0\}$

\Repeat
    \State $i \gets i+1$
    \State $\hat{\mathcal{T}}_{i-1} \gets \text{Empirical distributions from }\mathcal{T}_{i-1}$ using Alg.~\ref{alg:twisted-smc}
    \State $(\pi_{i-1}^*, \hat{\sigma}_{i-1}^*) \gets \textsc{MixedNashSolver}(\mathcal{X}_{i-1},\hat{\mathcal{T}}_{i-1}, U)$
    
    \State $\mathbf{x}_i \gets \underset{\mathbf{x}\in \mathcal{X}}{\arg\max}\;U(\mathbf{x}, \hat{\sigma}_{i-1}^*)$ //Adversary Oracle
    
    \State $\tau_i(\mathbf{z}) \propto p_{\theta}(\mathbf{z}|\mathbf{c}) 
            \exp\bigl(-\gamma\,U(\pi_{i-1}^*, \mathbf{z})\bigr)$ // Defender Oracle
    
    \State$\mathcal{X}_i \gets \mathcal{X}_{i-1} \cup \{\mathbf{x}_i\}, 
            \quad \mathcal{T}_i \gets \mathcal{T}_{i-1} \cup \{\tau_i\}$
    
  \State $\hat{\tau}_i \gets \text{sample from }\tau_i$ using Alg~\ref{alg:twisted-smc}
   \State $\underline{v}_i \gets U(\pi_{i-1}^*, \hat{\tau}_i), \quad
            \bar{v}_i \gets U(\mathbf{x}_i, \hat{\sigma}_{i-1}^*)$

\Until{$(\bar{v}_i - \underline{v}_i \in (-2\epsilon, 2\epsilon)) 
       \,\wedge\, (i > 1/(16\,\mathsf{prob}))$}

\State \textbf{Output:} Final defender strategy $\pi_{i-1}^*$
\end{algorithmic}
\end{algorithm}

\subsection{Sampling with Twisted Sequential Monte Carlo}\label{sec:twisted_sampler}

To efficiently sample from the unnormalized distribution in Eq.~\ref{eq:bo-poacher} while ensuring correctness, we leverage Twisted Sequential Monte Carlo (Twisted SMC)~\citep{chopin2020introduction}, an adaptive importance sampling technique that improves sampling through sequential proposal and weighting. \citet{wu2023practical} applied it to sampling from a conditional distribution with diffusion model; here, we adapt it to sample from the unnormalized reweighted distribution in Eq.~\ref{eq:bo-poacher}.

Twisted SMC operates with a collection of \(N\) weighted particles \(\{(w_n^t, \mathbf{z}_n^t)\}_{n=1}^N\) that evolve iteratively over \(T\) steps. At each step \(t\), particles are propagated using an adjusted score function, similar to \citet{chung2023diffusion}:
\[
\hat{p}_\theta(\mathbf{z}^{t-1} \mid \mathbf{z}^{t}, \mathbf{c}) 
= \mathcal{N}\!\bigl(\mathbf{z}^{t-1};\, \mathbf{z}^{t} + \sigma^2 \hat{s}_\theta(\mathbf{z}^{t}, \mathbf{c}, t),\, \hat{\beta}^2\bigr),
\]
where the adjusted score function is:
\[
\textstyle \hat{s}_\theta(\mathbf{z}^{t}, \mathbf{c}, t) 
= s_\theta(\mathbf{z}^{t}, \mathbf{c}, t) 
+ \gamma \log\Phi_t(\mathbf{z}^t).
\]
The twisting function \(\Phi_t\) is defined as:
\begin{align}
\textstyle\Phi_t(\mathbf{z}_n^t) = \exp\left(-\gamma  U(\pi_{i-1}^*, \hat{\mathbf{z}}_{\theta}^0(\mathbf{z}_n^t))\right).
\label{eq:twisting}
\end{align}
Here, \(\hat{\mathbf{z}}^0_\theta(\mathbf{z}^t)\) estimates the original state \(\mathbf{z}^0\) using Tweedie’s formula~\citep{robbins1992empirical,efron2011tweedie}:
\[
\textstyle
\hat{\mathbf{z}}^0_\theta(\mathbf{z}^t) = \mathbf{z}^t + t \beta^2\, s_\theta(\mathbf{z}^t, \mathbf{c}, t).
\]
At \(t = 0\), we set \(\hat{\mathbf{z}}^0_\theta(\mathbf{z}^0) \coloneqq \mathbf{z}^0\). The correction term in \(\hat{s}_\theta\) reconstructs \(\mathbf{z}^0\) and incorporates the reweighted term from Eq.~\ref{eq:bo-poacher}, ensuring proper adaptation of the sampling process.

To account for discrepancies between the proposal and target distributions, Twisted SMC assigns a weight to each particle:
\[
\textstyle w^t_n = \frac{p_\theta(\mathbf{x}_n^t|\mathbf{x}_n^{t+1},\mathbf{c}) \Phi_t(\mathbf{x}_n^t)}{ \hat{p}_\theta(\mathbf{x}_n^t|\mathbf{x}_n^{t+1},\mathbf{c})\Phi_{t+1}(\mathbf{x}_n^{t+1})}.
\]

This reweighting step ensures unbiased estimation.

To mitigate variance and prevent particle degeneracy over long horizons, we apply multinomial resampling at each step based on normalized weights~\citep{douc2005comparison}. The final approximation of the target distribution is:
$\hat{\tau} = \sum_{n=1}^N\frac{w_n^0}{\sum_{n'=1}^N w_{n'}^0} \delta_{\mathbf{z}_n^0}.$

A full description of Twisted SMC is provided in Algorithm~\ref{alg:twisted-smc}.


\begin{algorithm}[t]
\small
\caption{Twisted SMC for Diffusion Model}
\label{alg:twisted-smc}
\begin{algorithmic}[1]
\Require Pretrained diffusion model, number of particles $N$, time horizon $T$, $\Phi(\mathbf{z})$ (Eq.~\ref{eq:twisting})
\State Initialize $\mathbf{z}_n^T \sim p_{\theta}(\mathbf{z}^T)$,\; $w_n \gets \Phi(\mathbf{z}_n^T)$
\For{$t = T, \dots, 1$}
  \State \textbf{Resample:} \\
    $\quad\quad\{\mathbf{z}_n^t\}_{n=1}^N \sim \mathrm{Multinomial}\bigl(\{\mathbf{z}_n^t\}_{n=1}^N;\,\{w_n^t\}_{n=1}^N\bigr)$
  \For{$k = 1 \dots K$}
    \State $\displaystyle 
           \hat{s}_k \gets s_\theta(\mathbf{z}_k^t,\mathbf{c},t) \;-\; 
             \gamma \,\nabla_{\mathbf{z}_k^t}\Bigl[U(\pi^*_{i-1},\mathbf{z})\Bigr]$
    \State $\mathbf{z}_k^{t-1} \sim 
            \mathcal{N}\bigl(\mathbf{z}_k^t + \sigma^2 \hat{s}_k,\;\hat{\beta}^2\bigr)$
    \State $\displaystyle
           w_k^{t-1} \gets 
           \frac{p_{\theta}\bigl(\mathbf{z}_k^{t-1} \mid \mathbf{z}_k^t,\mathbf{c}\bigr)\,\Phi(\mathbf{z}_k^{t-1})}{
                 \hat{p}_{\theta}\bigl(\mathbf{z}_k^{t-1} \mid \mathbf{z}_k^t,\mathbf{c}\bigr)\,\Phi(\mathbf{z}_k^{t})}$
  \EndFor
\EndFor
\State \textbf{Output:} Weighted particles $\{\mathbf{z}_k^0,\, w_k^0\}_{k=1}^K$
\end{algorithmic}
\end{algorithm}

\begin{proposition}
\label{thm:twisted_diffusion_sampler}
(Informal) Under regularity conditions on the score function, as the number of particles \( N \to \infty \), we have
\[
U(\mathbf{x},\hat{\tau}(\mathbf{z})) \to U(\mathbf{x},\tau(\mathbf{z})) \quad \text{almost surely},
\]
where \( \hat{\tau} \) is the empirical distribution returned by Algorithm~\ref{alg:twisted-smc}.

Proof. See Appendix.~\ref{appdx:twisted_diffusion_model}.

\end{proposition}

\subsection{Convergence Analysis}\label{sec:convergence_analysis}

In Section \ref{sec:convergence_analysis}, we analyze the convergence properties of our framework. For theoretical analysis, we introduce two mild assumptions. 
\begin{assumption}\label{assump:concavity}
    We assume that the utility function is twice differentiable and concave with respect to $\mathbf{x}$.
\end{assumption}
Assumption \ref{assump:concavity} implies there is diminishing marginal return in ranger effort, which is  a common assumption in economics models \citep{mankiw1998principles} and reflects the intuition that initial patrol efforts contribute more significantly to wildlife protection than additional increments in effort. Under assumption \ref{assump:concavity}, Eq. \ref{eq:defender_oracle} is a convex optimization problem and existing optimization solvers~\citep{diamond2016cvxpy} can accurately find the defender's best response. 

\begin{assumption}
\label{assump:full_support}
We assume that the distribution \(p_{\theta}(\mathbf{z}\mid \mathbf{c})\) places its mass on a compact space.
\end{assumption}

 In practice, the attacker’s action at each target must lie in a bounded interval, e.g.\ \([0, z_{\mathrm{max}}]\). For instance, the number of snares at any region cannot exceed a practical upper limit. Consequently, it is reasonable to treat the action space as compact, ensuring that \(p_{\theta}(\mathbf{z}\mid \mathbf{c})\) has compact support.

For each $\hat{\sigma}^*_i$, we denote the corresponding mixed strategy on the underlying true adversary strategy distribution as $\sigma^*_i$. Formally, $\sigma^*_i(\tau_l) = \hat{\sigma}^*_i(\hat{\tau_l}) \ \forall l \in [i]$.
Without the terminating condition, Algorithm~\ref{alg:double-oracle} produces two sequences of mixed strategies: $(\pi_i^*)_{i=0}^{\infty}$ and $(\sigma_i^*)_{i=0}^{\infty}$. Proposition \ref{thm:twisted_diffusion_sampler} says if we use infinite samples to estimate expected utilities, then there is no estimation error and 
Theorem~\ref{cor:double_oracle_infinite} follows from the original double oracle algorithm's proof~\citep{adam2021double}.


\begin{theorem}\label{cor:double_oracle_infinite}
Without terminating conditions, under assumptions \ref{assump:concavity}, \ref{assump:full_support}, if we use $N \rightarrow \infty$ samples for all iterations, every weakly convergent subsequence of Alg.~\ref{alg:double-oracle} converges to an exact equilibrium in possibly infinite iterations. Such a weakly convergent subsequence always exists. \footnote{We include the definition of weak convergence in Appendix \ref{appdx:weak_convergence}.}
\end{theorem}

However, in practical scenarios where only a finite number of samples is available, the estimation of the expected utility is imprecise. Consequently, estimation errors will appear in the following steps within each iteration of our algorithm: (1) solving the subgame, (2) computing the defender oracle, and (3) evaluating the terminating condition.

\begin{restatable}{theorem}{convergenceinfround}\label{thm:convergence_inf_round}
Under assumptions \ref{assump:concavity} and \ref{assump:full_support}, with finite number of samples at the $i$-th iteration 
 $$N_i = \left\lceil CM^2(i+1)^2 i^{1+\delta}/\epsilon^2 \right\rceil, $$
for each adversary distribution, where $C$ is a constant, $M$ is the upper bound of utility function, $\epsilon$ is the approximation error, and $\delta$ is any positive number. 
\begin{itemize}
    \item \textbf{Item 1:} Without terminating condition, every weakly convergent subsequence of Alg.~\ref{alg:double-oracle} converges to an $\epsilon$-equilibrium in a possibly infinite number of iterations almost surely. Such a weakly convergent subsequence always exist. 
    \item \textbf{Item 2:} With the terminating condition, Alg.~\ref{alg:double-oracle} terminates in a finite number of iterations almost surely. Also, it converges to a finitely supported $4\epsilon$-equilibrium with probability at least $1-\mathsf{prob}$.
\end{itemize}
\end{restatable}
\begin{proof}
 We provide a sketch of the proof here and defer the full details to Appendix~\ref{appdx:convergence_inf_round}. The key steps for proving Item~1 are as follows:

\begin{itemize}
    \item \textbf{Step 1:} We bound the utility estimation error for any mixed strategy pair at iteration \(i\) by the maximum estimation error over all entries in the payoff matrix.
    \item \textbf{Step 2:} We show that, under our finite sampling scheme, almost surely, the realized sequence of payoff matrices
    has a finite cutoff index after which all subsequent iterations have maximum cell-wise errors below $\epsilon/4$.
    \item \textbf{Step 3:} 
    We treat the strategies generated during the finite initial phase (before the cutoff index ensured by Step 2) as the initial strategy set in the standard Double Oracle (DO) algorithm~\citep{adam2021double}. We then adapt the original convergence proof to account for the error introduced by finite sampling, which is now bounded by \(\epsilon/4\).
\end{itemize}

By relaxing the error bound in Item~1, we obtain convergence within a finite number of iterations almost surely. The additional approximation error in Item~2 stems from two sources: (1) enforcing finite termination, and (2) using estimated utilities of mixed strategy pairs when evaluating the stopping condition.

\end{proof}
In practice, we use a fixed number of samples across iterations, and experiments in Section~\ref{sec:experiments} shows our framework still achieves robust performance.

\begin{figure*}
    \centering
    \includegraphics[width=0.95\linewidth]{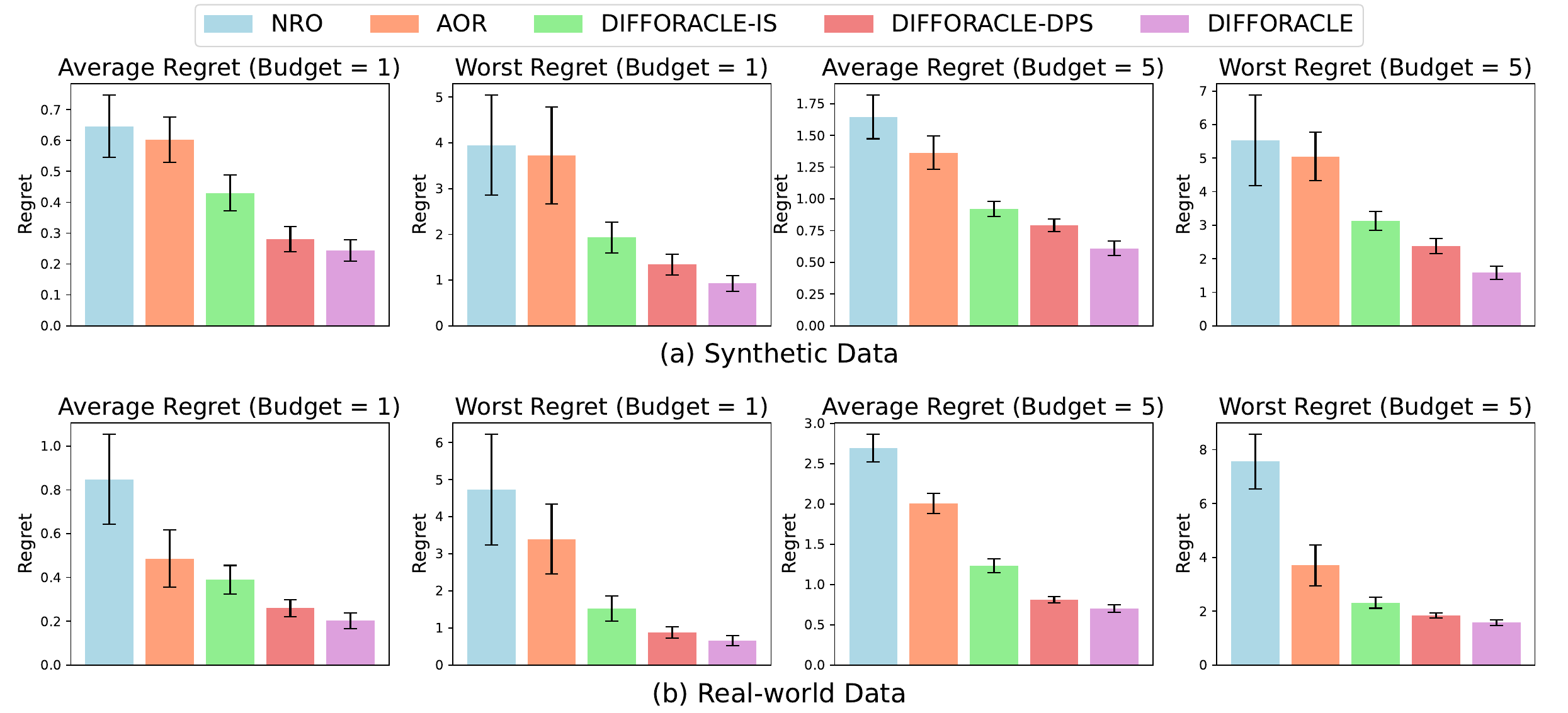}
    \caption{Experimental Results on both synthetic and real-world datasets. Following \citep{ho2020denoising}, we average the results over 5 random seeds.}
    \label{fig:exp}
\end{figure*}

\section{Experiments}\label{sec:experiments}

\subsection{Experimental Setup}

\paragraph{Datasets}
We conduct experiments on both synthetic and
real-world datasets which we describe below. We use
a graph based dataset~\citep{nguyen2016capture} to reflect geospatial constraints
in the poaching domain for patrollers.

 \textbf{Synthetic data.} 
Poaching counts are sampled from a Gamma distribution parameterized by shape and scale values. To determine the shape parameter, we randomly select one of two Graph Convolutional Networks (GCNs)~\citep{kipf2022semi} with randomly initialized weights to map the node’s feature vector to a continuous value, which serves as the shape parameter. The scale parameter is set to 1 if the first GCN is chosen and 0.9 if the second is selected. Finally, adversarial noise inversely proportional to the poaching count is added, ensuring that nodes with lower poaching counts receive higher noise levels.


\textbf{Real-world Data.} We use poaching data from Murchison Falls National Park (MFNP) in Uganda, collected between 2010 and 2021. The protected area is discretized into 1 × 1 km grid cells. For each cell, we measure ranger patrol effort (in kilometers patrolled) as the conditional variable for the diffusion model, while the monthly number of detected illegal activity instances serves as the adversarial behavior. Following~\citet{basak2016abstraction}, we represent the park as a graph to capture geospatial connectivity among these cells. To focus on high-risk regions, we subsample 20 subgraphs from the entire graph. Specifically, at each month we identify the 20 cells with the highest poaching counts. Each of these cells is treated as a central node, and we iteratively add the neighboring cell with the highest poaching count until the subgraph reaches 20 nodes. This process generates 532 training, 62 validation, and 31 test samples.

\textbf{Baselines}
We compare the following methods:

\textbf{Non-robust Optimization (NRO).} We use a baseline that directly maximizes the expected utility under the pre-trained diffusion model. The stochastic optimization is solved via sample average approximation, using samples from the diffusion model in conjunction with mirror ascent. Since this approach yields only a pure strategy, we repeat the process with different initializations to obtain five distinct pure strategies. These pure strategies are then combined into a mixed strategy by assigning them equal probability.  

\textbf{Alternate Optimization with Random Reinitialization (AOR).} This method solves the DRO problem using alternating optimization without employing the double oracle framework. It iterates between optimizing the defender's strategy and sampling from the worst-case distribution using twisted SMC. Similar to NRO, we construct a mixed strategy by running the procedure five times with different initializations, generating multiple pure strategies that are then combined with equal probability.  

We also compare against three variations of \ours:  

\textbf{\ours with Importance Sampling (\ours-IS).} This variant replaces the twisted SMC sampler in Section~\ref{sec:twisted_sampler} with importance sampling to sample from Eq.~\ref{eq:bo-poacher}. As the proposal distribution, we directly use the pre-trained diffusion model, $p_{\theta}(\mathbf{z}|\mathbf{c})$.  

\textbf{\ours with Diffusion Posterior Sampling (\ours-DPS).} This variant employs the diffusion posterior sampler~\citep{chung2023diffusion} instead of the twisted SMC sampler in Section~\ref{sec:twisted_sampler}.   

\textbf{\ours.} This version retains the default twisted SMC sampler in Section~\ref{sec:twisted_sampler} to sample from Eq.~\ref{eq:bo-poacher}.

\textbf{Evaluation metrics}
We evaluate the methods using decision \textit{regret}, defined as the difference between the defender's best possible utility under the true adversarial behavior and the expected utility under the optimized mixed strategy. We report both the average regret on the test set and the worst-case regret on the test set.

\textbf{Implementation details}
We employ a three-layer GCN with a hidden dimension of 128 as the backbone of the diffusion model. The optimizer used is Adam~\citep{kingma2014adam} with a learning rate of $10^{-3}$. We use $500$ samples to estimate the expected utility for all the menthods. $\gamma$ is selected on the validation set based on the average regret. More details of the implementation details are provided in Appendix~\ref{sec:exp-details}.


\subsection{Experimental Results}

\paragraph{Main Results.} We evaluate our method against baselines on both synthetic and real-world poaching datasets under different patrol budgets (\(B=1\) and \(B=5\)). The results, presented in Fig.~\ref{fig:exp}, show that \ours consistently achieves the lowest average regret and worst-case regret across all settings.  

Compared to NRO, \ours reduces average regret by \(62.2\%\), \(62.9\%\), \(73.3\%\), and \(74.0\%\) across different datasets and budgets. Similarly, worst-case regret is significantly reduced by \(59.1\%\), \(64.9\%\), \(71.3\%\), and \(79.0\%\). These improvements highlight the robustness of our approach, which is particularly crucial in green security domains, where minimizing worst-case outcomes is essential for high-stakes decision-making.  

The double-oracle framework provides substantial benefits, as all three variants of \ours significantly outperform the naive robust optimization approach, AOR. This is because AOR relies on a simple heuristic to solve the minimax problem and construct the mixed strategy, lacking convergence guarantees. Consequently, AOR exhibits greater variance and instability, further underscoring the advantages of employing game-theoretic methods for robust optimization. 


Among the \ours methods employing different sampling strategies, DPS emerges as the strongest alternative to twisted SMC. However, SMC consistently outperforms DPS, demonstrating statistically significant improvements in five out of eight cases. Furthermore, DPS cannot exactly sample from the target distribution in Eq.\ref{eq:bo-poacher}~\citep{lu2023contrastive}, a critical requirement for ensuring the convergence of the double-oracle framework, as analyzed in Section~\ref{sec:convergence_analysis}.

\begin{figure}
    \centering
    \includegraphics[width=\linewidth]{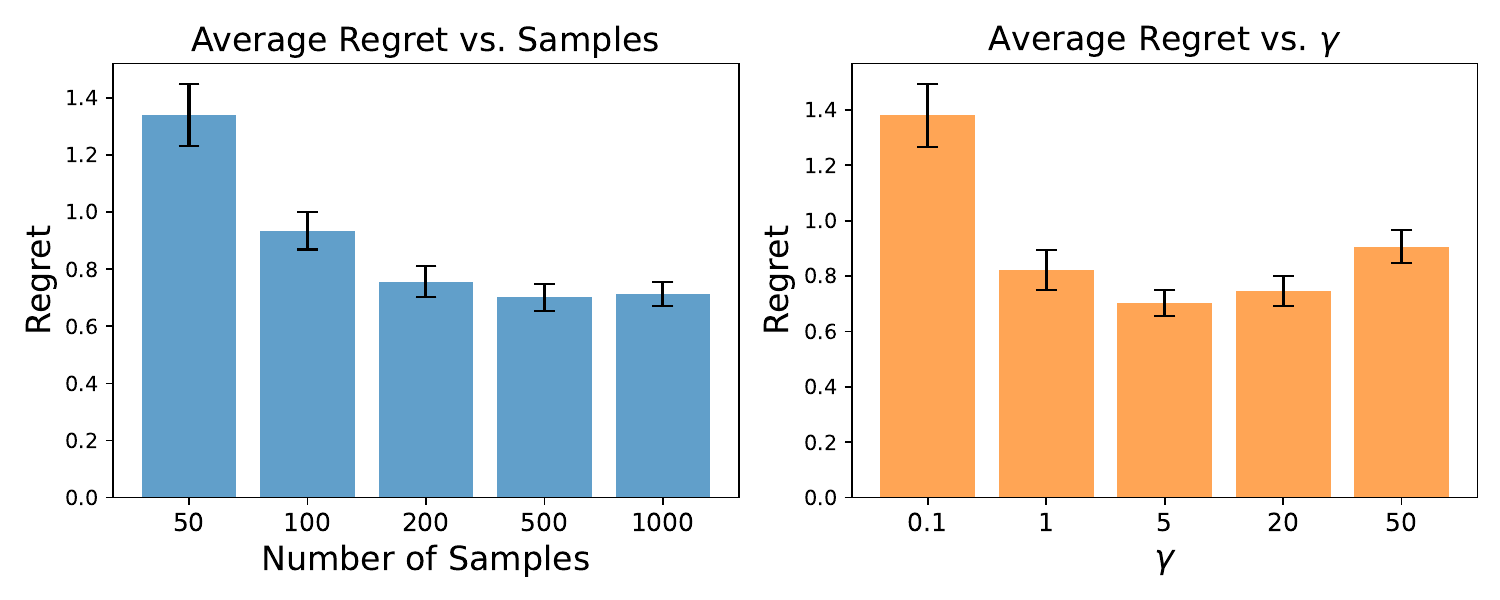}
    \caption{Parameter Study on \ours-SMC on poaching data under Budget 5.}
    \label{fig:parameter}
\end{figure}

\textbf{Parameter Study.} Fig.~\ref{fig:parameter} presents the parameter study of \ours using the twisted SMC sampler. As shown, varying the number of samples in the sampler reveals that once the sample size exceeds 200, performance stabilizes. Additionally, when adjusting the value of $\gamma$, we observe that performance drops significantly as $\gamma$ approaches 0, since it effectively reduces to non-robust optimization. Conversely, when $\gamma$ is too large, performance also declines because the nature adversary may select a worst-case distribution that deviates too far from the learned distribution, making it non-informative.

\section{Conclusion}

We introduced a conditional diffusion model for adversary behavior modeling in green security, overcoming the limitations of traditional Gaussian process and linear models. To the best of our knowledge, this is the first application of diffusion models in this domain.  To integrate diffusion models into game-theoretic optimization, we proposed a mixed strategy of mixed strategies and leverage a twisted Sequential Monte Carlo (SMC) sampler for efficient sampling from unnormalized distributions. We established theoretical convergence to an \(\epsilon\)-equilibrium with high probability using finite samples and finite iterations and demonstrated empirical effectiveness on both synthetic and real-world poaching datasets.  
 Future work could explore extensions to sequential-decision-making.

\section*{Acknowledgements}
We are thankful to the Uganda Wildlife Authority for granting us access to incident data from Murchison Falls National Park.
We also thank Charles Emogor for the insightful discussions and the anonymous reviewers for their valuable feedback. This work was supported by ONR MURI N00014-24-1-2742.

\bibliography{references}


\newpage
\appendix
\onecolumn

\begin{center}
	{\Large \textbf{Appendix for Robust Optimization with Diffusion Models for Green Security}}
\end{center}

\startcontents[sections]
\printcontents[sections]{l}{1}{\setcounter{tocdepth}{2}}

\section{More Details on Diffusion Models}

A diffusion model \citep{sohl2015deep} is a generative framework composed of two stochastic processes: a \emph{forward} process that progressively adds Gaussian noise to real data, and a \emph{reverse} (or denoising) process that learns to remove this noise step by step. Formally, let \(\mathbf{z}^0 \sim \mathcal{D}\) be a sample from the training dataset.\footnote{We use \(\mathbf{z}^0\) and \(\mathbf{z}\) interchangeably when there is no ambiguity.} The forward diffusion process can be written as
$q(\mathbf{z}^t \mid \mathbf{z}^{t-1}) 
= \mathcal{N}\!\bigl(\mathbf{z}^t;\,\mathbf{z}^{t-1},\,\beta^2 \mathbf{I}\bigr),$
where \(\beta^2\) is the noise variance at each step \(t=1,\dots,T\). As \(T\) becomes large, repeated noising transforms the data distribution into (approximately) pure Gaussian noise:
$q(\mathbf{z}^T) \approx \mathcal{N}(\mathbf{0},\,T \beta^2 \mathbf{I}).$

\textbf{Score-based Approximation.} To invert this process (i.e., to denoise and recover samples from the original data distribution), one can approximate the reverse transition 
\(
q(\mathbf{z}^{t-1} \mid \mathbf{z}^t) 
\)
via the \emph{score function}, \(\nabla_{\mathbf{z}^t} \log q(\mathbf{z}^t)\) when $\beta$ is small. Specifically,
\[
q(\mathbf{z}^{t-1} \mid \mathbf{z}^t) 
\,\approx\, 
\mathcal{N}\!\Bigl(\mathbf{z}^{t-1};\,
\mathbf{z}^t 
+ \beta^2 \,\nabla_{\mathbf{z}^t}\! \log q(\mathbf{z}^t),\,
\beta^2 \mathbf{I}\Bigr).
\]
Here, \(q(\mathbf{z}^t) = \int q(\mathbf{z}^0)\,q(\mathbf{z}^t \mid \mathbf{z}^0)\, d\mathbf{z}^0\), and the gradient \(\nabla_{\mathbf{z}^t}\! \log q(\mathbf{z}^t)\) points toward regions of higher data density. In practice, we do not know \(q(\mathbf{z}^t)\) in closed form, so a neural \emph{score network} \(s_{\theta}(\mathbf{z}^t, t)\) is trained to approximate this gradient via \emph{denoising score matching} \citep{vincent2011connection, ho2020denoising}. Consequently, the learned reverse (denoising) transition becomes
\[
p_{\theta}(\mathbf{z}^{t-1} \mid \mathbf{z}^t) 
= \mathcal{N}\!\Bigl(\mathbf{z}^{t-1};\,
\mathbf{z}^t 
+ \beta^2\, s_{\theta}(\mathbf{z}^t, t),\,
\beta^2 \mathbf{I}\Bigr).
\]
Starting from an initial Gaussian sample \(\mathbf{z}^T \sim \mathcal{N}(\mathbf{0},\, T \beta^2 \mathbf{I})\), iterating this reverse process ultimately recovers samples that approximate the original data distribution.

\textbf{Conditional Extension.} This diffusion framework can be naturally extended to include additional context \(\mathbf{c}\). In a \emph{conditional} diffusion model~\citep{ho2021classifier}, the score network becomes \(s_{\theta}(\mathbf{z}^t, t, \mathbf{c})\), so that at each step the denoising is informed by side information such as class labels, textual descriptions, or other relevant features. This conditional approach enables the generation of samples that match not only the learned data distribution but also the specific context \(\mathbf{c}\), making it particularly useful for tasks in which external conditions strongly influence the underlying data generation process.

Rather than directly estimating the score function \( s_{\theta}(\mathbf{z}^t, t) \), Denoising Diffusion Probabilistic Models (DDPM)~\citep{ho2020denoising} reformulate the learning objective as a \textit{noise prediction} task. This reparameterization leverages the closed-form expression of the forward process:
\[
\mathbf{z}^t = \sqrt{\bar{\alpha}_t} \mathbf{z}^0 + \sqrt{1 - \bar{\alpha}_t} \boldsymbol{\epsilon}, \quad \boldsymbol{\epsilon} \sim \mathcal{N}(0, \mathbf{I}),
\]
where \(\bar{\alpha}_t\) denotes the cumulative product of noise schedules. The training objective becomes recovering the noise \(\boldsymbol{\epsilon}\) that perturbed \(\mathbf{z}^0\) to form \(\mathbf{z}^t\). A neural network \(\boldsymbol{\epsilon}_\theta(\mathbf{z}^t, t)\) is trained to approximate this noise, which corresponds to learning the score function up to a time-dependent scaling:
\[
s_{\theta}(\mathbf{z}^t, t) = - \frac{\boldsymbol{\epsilon}_\theta(\mathbf{z}^t, t)}{\sqrt{1 - \bar{\alpha}_t}}.
\]

Training then reduces to minimizing a simple mean squared error (MSE) loss between the true and predicted noise:
\[
\mathcal{L}_{\text{simple}} = \mathbb{E}_{\mathbf{z}_0,\, \boldsymbol{\epsilon},\, t}\!
  \bigl[\|\boldsymbol{\epsilon} 
    - \boldsymbol{\epsilon}_\theta(\mathbf{z}_t, t, \mathbf{c})\|^2\bigr].
\]

By training this conditional diffusion model on historical poaching data---augmented with contextual features \(\mathbf{c}\)---we learn \( p_\theta(\mathbf{z} \mid \mathbf{c}) \), a powerful and expressive model of poacher behavior. This enables us to capture complex, multimodal patterns of attacker responses, thereby supporting the development of robust patrol strategies discussed earlier. One challenge of this task is that poaching data is often very noisy~\citep{zhuang2023dygen}.

\section{Examples of mixed 
strategy over mixed strategies}\label{appdx:example}
Let us consider a national park with $3$ target regions to protect, and poachers' pure strategies specify how many snares to put in each target region. Two examples of poacher pure strategies could be $\mathbf{z}_1 = (3, 4, 3)$ and $\mathbf{z}_2 = (0, 0, 10)$. Each entry in the pure strategy determines the number of snares a poacher will place in the corresponding target region. Let us denote poachers' pure strategy space as $\mathcal{Z} = \{\mathbf{z}_1, \mathbf{z}_2\}$. 

A mixed strategy $\tau$ is a distribution on the pure strategy space, i.e., $\tau \in \Delta(\mathcal{Z})$. Denote the subset of mixed strategies which satisfy the constraint 
$D_{\rm KL}(\tau(\mathbf{z}) || p_{\theta}(\mathbf{z} | \mathbf{c})) \leq \rho$ as $\mathcal{T}$. One such example $\tau_1$ could be $P(\mathbf{z}_1) = 0.1$ and $P(\mathbf{z}_2) = 0.9$. Another degenerate example of mixed strategy $\tau_2$ could be $P(\mathbf{z}_1) = 0$ and $P(\mathbf{z}_2) = 1$.

A mixed strategy over mixed strategies $\sigma$ is a distribution on the constrained mixed strategy space, i.e., $\sigma \in \Delta(\mathcal{T})$. One example of mixed strategy over mixed strategies $\sigma_1$ could be $P(\tau_1) = 0.1$ and $P(\tau_2) = 0.9$. Another degenerate example $\sigma_2$ could be $P(\tau_1) = 0$ and $P(\tau_2) = 1$.

A mixed strategy over mixed strategies is still a distribution on the original pure strategy space, i.e., $\sigma \in \Delta(\mathcal{Z})$. For example, an alternative way to view $\sigma_1$ could be $$P(\mathbf{z}_1) = P(\sigma_1(\tau_1)) \cdot P(\tau_1(\mathbf{z}_1)) + P(\sigma_1(\tau_2)) \cdot P(\tau_2(\mathbf{z}_1)) = 0.01$$ 
and 
$$P(\mathbf{z}_2) = P(\sigma_1(\tau_1)) \cdot P(\tau_1(\mathbf{z}_2)) + P(\sigma_1(\tau_2)) \cdot P(\tau_2(\mathbf{z}_2)) = 0.99$$
However, it is proven in Proposition $\ref{thm:mixed_over_mixed}$ that all mixed strategy over mixed strategies $\sigma$ satisfy $D_{\rm KL}(\sigma || p_{\theta}(\mathbf{z} | \mathbf{c})) \leq \rho$, which is not generally true for elements in $\Delta(\mathcal{Z})$.

From Section $\ref{sec:double_oracle}$ onward, readers can interpret $\mathcal{T}$ as the pure strategy space and $\sigma$ as a standard mixed strategy. Despite each pure strategy $\tau \in \mathcal{T}$ being a distribution, 
all standard terminologies of game theory remain applicable.


\section{Proof of Proposition \ref{thm:mixed_over_mixed}}\label{appdx:mixed_over_mixed}

We now show that for any $\pi(\mathbf{x})\in \Delta(\mathcal{X})$, 

\begin{align*}
\min_{\tau(\mathbf{z})} \left\{ \mathbb{E}_{\pi(\mathbf{x})}\mathbb{E}_{\tau(\mathbf{z})} \left[ u(\mathbf{x}, \mathbf{z}) \right] : D_{\mathrm{KL}}(\tau(\mathbf{z}) \,\|\ p_{\theta}(\mathbf{z} | \mathbf{c})) \leq \rho \right\}=\\
\min_{\sigma(\tau)} \left\{ \mathbb{E}_{\pi(\mathbf{x})}\mathbb{E}_{\sigma(\tau)} \left(\mathbb{E}_{\tau(\mathbf{z})}\left[ u(\mathbf{x}, \mathbf{z}) \right] \right): D_{\mathrm{KL}}(\tau(\mathbf{z}) \,\|\,  p_{\theta}(\mathbf{z} | \mathbf{c})) \leq \rho \right\}.
\end{align*}
From this, the original theorem follows.

Consider any solution $\tau'(\mathbf{z})$ that attains the minimum on the left-hand side. Define a degenerate distribution over strategies $\sigma'(\tau) = \delta[\tau = \tau']$, i.e., it places all its mass on $\tau'$. Note that $\tau'$ satisfies the divergence constraint on the left, so $\sigma'(\tau)$ will also satisfy the corresponding constraint on the right-hand side. Since the expected value under $\sigma'(\tau)$ matches the value attained by $\tau'$, we have the left side is not smaller than the right side.

Now take any solution $\sigma'(\tau)$ that attains the minimum on the right side. Define
$\tau'(\mathbf{z}) = \mathbb{E}_{\sigma'(\tau)}[\tau(\mathbf{z})]$. 
Because a mixture over mixed strategies is itself a valid mixed strategy in $\Delta(\mathcal{Z})$, $\tau'(\mathbf{z})$ is admissible on the left side.

By the convexity of the divergence measure $D$, we have:
\[
D_{\mathrm{KL}}(\tau'(\mathbf{z})\,\|\ p_\theta(\mathbf{z}| \mathbf{c})) 
= D_{\mathrm{KL}}\bigl(\mathbb{E}_{\sigma'(\tau)}\tau(\mathbf{z}) \,\|\  p_\theta(\mathbf{z}| \mathbf{c})\bigr)
\leq \mathbb{E}_{\sigma'(\tau)}[D_{\mathrm{KL}}(\tau(\mathbf{z}) \,\|\ p_\theta(\mathbf{z}|\mathbf{c})]
\leq \rho.
\]
Here, the first inequality follows from the convexity of $D$, and the second inequality is by the construction of $\sigma'(\tau)$, which satisfies the original constraint on the right side.

Thus, $\tau'(\mathbf{z})$ satisfies the left side constraint and attains the same expected value as $\sigma'(\tau)$. We then obtain that the left side is not larger than the right side.

Combining both parts, we conclude the proof. 

\section{Proof of Proposition~\ref{propo:kl}}\label{appdx:propo_kl}

\begin{proof}
We introduce a Lagrange multiplier \(\alpha\) for the KL-divergence constraint and another multiplier \(\lambda\) for the normalization constraint. The Lagrangian is
\[
\mathcal{L}(\tau,\alpha,\lambda)\;=\;
\int \tau(\mathbf{z}) \Bigl(U(\pi_{i-1}^*, \mathbf{z})\Bigr)\,d\mathbf{z}
\;-\;\alpha \Bigl(D_{\rm KL}(\tau(\mathbf{z})||p_{\theta}(\mathbf{z}|\mathbf{c})) - \rho\Bigr)
\;+\;\lambda \Bigl(\!\int \tau(\mathbf{z})\,d\mathbf{z}-1\Bigr).
\]
By taking the functional derivative of \(\mathcal{L}\) with respect to \(\tau(\mathbf{z})\) and setting it to zero, one obtains
\[
\tau(\mathbf{z})
\;\propto\;
p_{\theta}(\mathbf{z}\mid \mathbf{c})
\exp\!\Bigl(\tfrac{1}{\alpha}\, U(\pi_{i-1}^*,\mathbf{z}))\Bigr).
\]
Defining \(\gamma = -\tfrac{1}{\alpha}\) (where \(\gamma>0\) absorbs constants and signs from the Lagrange approach) gives the closed-form solution
\[
\tau_i(\mathbf{z})
\;\propto\;
p_{\theta}(\mathbf{z}\mid \mathbf{c})
\exp\!\Bigl(-\gamma U(\pi_{i-1}^*,\mathbf{z})\Bigr),
\]
which matches Eq.~\eqref{eq:bo-poacher}. This completes the proof.
\end{proof}

\renewcommand{\theproposition}{5.3}

\section{Proof of Proposition~\ref{thm:twisted_diffusion_sampler}}\label{appdx:twisted_diffusion_model}

We first provide the full statement of Proposition~\ref{thm:twisted_diffusion_sampler} as below.
\begin{proposition}[Twisted SMC] 
Suppose the following conditions hold:
\begin{enumerate}
    \item \(\Phi_t(\mathbf{z}^T)\) and \(\Phi_{t}(\mathbf{z}^t)/\Phi_{t-1}(\mathbf{z}^{t-1})\) are positive and bounded.
    \item For \( t > 0 \), \( \log \Phi_t(\mathbf{z}^t) \) is continuous and has bounded gradients with respect to \( \mathbf{z}^t \).
    \item \( \hat{\beta}^2 > \beta^2 \).
\end{enumerate}

\textbf{Almost Sure Convergence:}  Under these assumptions, as the number of particles \( N \to \infty \), we have
\[
U(\mathbf{x},\hat{\tau}(\mathbf{z})) \to U(\mathbf{x},\tau(\mathbf{z})) \quad \text{almost surely},
\]
where \( \hat{\tau} \) is the empirical distribution returned by Algorithm~\ref{alg:twisted-smc}.

\textbf{Error Bound under Finite Samples:}\label{finite_sample_bound} 
Under these assumptions, the mean squared error of the twisted SMC sampler satisfies the bound:
\[
\mathbb{E}\left[|U(\mathbf{x}, \tau(\mathbf{z})) - U(\mathbf{x}, \hat{\tau}(\mathbf{z}))|^2 \right] \leq \frac{C'M^2}{N},
\]
where $C'$ is a constant and $M$ is the maximum value of the utility function.
\end{proposition}

\textbf{Justification of Assumptions:}  
\begin{itemize}
    \item \textbf{Assumption (1):} This holds if \( \exp(-\gamma U(\pi, \mathbf{z})) \) is positive and bounded away from zero. In our green security domain, \( U \) is always positive (as introduced in Section~X), ensuring this condition is automatically satisfied.  
    \item \textbf{Assumption (2):} This is justified by the Appendix A.5 of \citet{wu2023practical}.  
    \item \textbf{Assumption (3):} This can be ensured by selecting a sufficiently large \( \hat{\beta} \).  
\end{itemize}

\begin{proof}
Recall that \( p_{\theta}(\mathbf{z}|\mathbf{c}) \) serves as the prior, while the likelihood is given by \( \exp\left(-\gamma U(\pi, \mathbf{z})\right) \).

We first prove that the marginal distribution of the sampler is \( \tau(\mathbf{z}) \):

\begin{align}
\hat{p}(\mathbf{z}^{0:T}) &= \frac{1}{Z} \left[ p_{\theta}(\mathbf{z}^T) \prod_{t=1}^{T-1} \hat{p}(\mathbf{z}^t|\mathbf{z}^{t-1}) \right] 
\left[ \Phi_T(\mathbf{z}^T) \prod_{t=1}^{T-1} \frac{p_{\theta}(\mathbf{z}^t|\mathbf{z}^{t+1}) \Phi_t(\mathbf{z}^t)}{\hat{p}_{\theta}(\mathbf{z}^t|\mathbf{z}^{t+1}) \Phi_{t+1}(\mathbf{z}^{t+1})} \right] \nonumber \\
&= \frac{1}{Z} \left[ p_{\theta}(\mathbf{z}^T) \prod_{t=1}^{T-1} p(\mathbf{z}^t|\mathbf{z}^{t-1}) \right] 
\left[ \Phi_T(\mathbf{z}^T) \prod_{t=1}^{T-1} \frac{\hat{p}_{\theta}(\mathbf{z}^t|\mathbf{z}^{t+1}) \Phi_t(\mathbf{z}^t)}{\hat{p}_{\theta}(\mathbf{z}^t|\mathbf{z}^{t+1}) \Phi_{t+1}(\mathbf{z}^{t+1})} \right] \nonumber \\
&= \frac{1}{Z} p_{\theta}(\mathbf{z}^{0:T}) 
\left[ \prod_{t=0}^{T-1} \frac{\Phi_t(\mathbf{z}^t)}{\Phi_{t+1}(\mathbf{z}^{t+1})} \right] \Phi_T(\mathbf{z}^T) \nonumber \\
&= \frac{1}{Z} p_{\theta}(\mathbf{z}^{0:T}) \Phi_0(\mathbf{z}^0).
\end{align}

Since \( \Phi_0(\mathbf{z}^0) = \exp(-\gamma U(\pi, \mathbf{z}^0)) \), marginalizing out \( \mathbf{z}^{1:T} \) yields

\[
\hat{p}(\mathbf{z}^0) = \tau(\mathbf{z}^0) \propto p_{\theta}(\mathbf{z}|\mathbf{c}) \exp\left( -\gamma U(\pi, \mathbf{z}) \right),
\]

as desired.

Next, according to Appendix A.5 in \citet{wu2023practical}, under the given assumptions, the importance weights \( w^t \) remain bounded. Consequently, applying Propositions 11.5 and 11.3 from \citet{chopin2020introduction} establishes almost sure convergence and the error bound under finite samples.

\end{proof}

\section{Definition of weak convergence}\label{appdx:weak_convergence}
We directly cite the definition of weak convergence provided in \cite{adam2021double}, and a more detailed discussion of the convergence of probability measures can be seen in \cite{billingsley2013convergence}.

\begin{definition}
    A sequence of mixed strategy $(\pi_i^*)_{i=1}^{\infty}$ in $\Delta(\mathcal{X})$ weakly converges to $\pi^* \in \Delta(\mathcal{X}) $ if
    $$\lim\limits_{i \to \infty}\int_{\mathcal{X}} f(x) d\pi_i = \int_{\mathcal{X}} f(x) d\pi^*$$
    for every continuous function $f: \mathcal{X} \to R$. We use $\pi_i  \Rightarrow \pi^*$ to denote weak convergence.
\end{definition}
\section{Proof of Theorem \ref{thm:convergence_inf_round}}\label{appdx:convergence_inf_round}



\convergenceinfround*

The constant $C$ in $N_i$ can be expressed as $16C'$, where $C'$ is the constant in \textbf{Error Bound under Finite Sample} discussed in Appendix \ref{appdx:twisted_diffusion_model}. To prove Theorem \ref{thm:convergence_inf_round}, we first prove the utility estimation error bound for the twisted diffusion sampler.
\begin{lemma}\label{lem:chebyshev}
    Under the same assumptions as Proposition \ref{thm:twisted_diffusion_sampler}, the utility estimation error of the twisted SMC sampler satisfies the bound:

\[
P(|U(\mathbf{x}, \tau(\mathbf{z}))-U(\mathbf{x},\hat{\tau}(\mathbf{z}))| \geq \epsilon) \leq \frac{C'M^2}{N\epsilon^2},
\]
where $C'$ is the constant in \textbf{Error Bound under Finite Sample} of Appendix \ref{appdx:twisted_diffusion_model} and $M$ is the maximum value of the utility function.
\end{lemma}
\begin{proof}
    Consider the random variable $U(\mathbf{x}, \tau(\mathbf{z}))-U(\mathbf{x},\hat{\tau}(\mathbf{z}))$, whose variance is upper bounded by $E|(U(\mathbf{x}, \tau(\mathbf{z}))-U(\mathbf{x},\hat{\tau}(\mathbf{z}))^2|$. By \textbf{Error Bound under Finite Sample}, we know that this variance is upper bounded by $\frac{C'M^2}{N}$.

    Applying Chebyshev's inequality to the random variable $U(\mathbf{x}, \tau(\mathbf{z}))-U(\mathbf{x},\hat{\tau}(\mathbf{z}))$, we have
    $$P(|U(\mathbf{x}, \tau(\mathbf{z}))-U(\mathbf{x},\hat{\tau}(\mathbf{z})| \geq \epsilon) \leq \frac{C'M^2}{N\epsilon^2}$$
\end{proof}

\paragraph{Notation} We introduce notations used in the proof of Theorem \ref{thm:convergence_inf_round}. Recall at the $i$-th iteration of algorithm \ref{alg:double-oracle}, we use an empirical distribution $\hat{\tau_i}$ with $N_i$ samples to approximate each adversary strategy (distribution) $\tau_i \in \mathcal{T}_i$.  Because of the finite sample approximation, the utility estimation for each pure strategy pair in the payoff matrix is imprecise. Define the estimation error in row $j$, column $k$ of payoff matrix at iteration $i$ as
$$\Delta_i^{j,k} = U(x_j, \tau_k) - U(x_j, \hat{\tau}_k).$$ 
Let $\Delta_i$ denote the absolute value of the largest utility estimation error for any cell in the payoff matrix at the $i$-th iteration of the algorithm, i.e., $\Delta_i = \max_{j,k} |\Delta_i^{j,k}|$. At step 7 of algorithm \ref{alg:double-oracle}, when we apply linear programming to solve the subgame $(\mathcal{X}_{i}, \hat{\mathcal{T}}_{i}, U)$, we obtain a mixed strategy for adversary $\hat{\sigma}^*_{i}$ defined on $\hat{\mathcal{T}}_{i}$. Recall $\sigma^*_i \in \Delta(\mathcal{T}_i)$ is the mixed strategy on the underlying true adversary distribution that shares the same weight as  $\hat{\sigma}^*_i$. Formally, $\sigma^*_i(\tau_l) = \hat{\sigma}^*_i(\hat{\tau_l}) \ \forall l \in [i]$.

\paragraph{Proof of Item 1}
\begin{proof}
We first show that for any $\pi_i \in \Delta(\mathcal{X}_{i})$ and $\sigma_i \in \Delta(\mathcal{T}_{i})$, we have $|U(\pi_i, \hat{\sigma}_i) - U(\pi_i, \sigma_i)| \leq \Delta_i$.
We write 
    \begin{equation}\label{eq:utility_bound_pf}
        |U(\pi_i, \sigma_i) - U(\pi_i, \hat{\sigma}_i)| = \sum_{\mathbf{x} \in \mathcal{X}_i}\sum_{\tau \in \mathcal{T}_i} \pi_i(\mathbf{x}) \cdot \sigma_i (\tau) \cdot |U(\mathbf{x}, \tau) - U(\mathbf{x}, \hat{\tau})|
    \end{equation}
   $|U(\mathbf{x}, \tau) - U(\mathbf{x}, \hat{\tau})|$ denotes the sample estimation error for the pure strategy pair $(\mathbf{x}, \tau)$ in the payoff matrix. The maximum on the right-hand side of Equation \ref{eq:utility_bound_pf} is obtained when putting all the probability mass on the strategy pair with the largest sample estimation error, which is $\Delta_i$.

    Then we bound $\Delta_i$ for each $i$. For any cell $(j,k)$ in the matrix, we apply Lemma \ref{lem:chebyshev}:
    $$P(|\Delta_{i}^{j,k}| \geq \frac{\epsilon}{4}) \leq \frac{16C'M^2}{N\epsilon^2 }.$$
    Since at $i$-th iteration, there are $(i+1)^2$ cells in the payoff matrix, we apply the union bound and obtain:
     $$P(\Delta_i \geq \frac{\epsilon}{4}) \leq \frac{16C'M^2 (i+1)^2}{N_i\epsilon^2}.$$
    By setting $N_i =\left\lceil 16C'M^2(i+1)^2 i^{1+\delta}/\epsilon^2 \right\rceil$, we have 
    $$P(\Delta_i \geq \frac{\epsilon}{4}) \leq \frac{1}{i^{1+\delta}}.$$
    Here we consider the events $A_i=\{\Delta_i \geq \frac{\epsilon}{4}\}$. From the step above, we have $$P(A_i)\leq \frac{1}{i^{1+\delta}}.$$ 
    Because $\delta >0$, $\sum_{i=1}^{\infty} \frac{1}{i^{1+\delta}}$ is a convergent series. Therefore, $\sum_{i=1}^{\infty} P(A_i)<\infty$. From Borel-Cantelli Lemma, we then obtain $P(\limsup\limits_{i \to \infty}A_i)=0$. Let's use $\Omega$ to denote the space of all payoff matrix sequence across iterations of Alg. \ref{alg:double-oracle} without terminating condition, and use $\omega \in \Omega $ to denote a realization. $P(\limsup\limits_{i \to \infty}A_i)=0$ implies that
    $$P(\exists i_r(\omega) \text{ finite s.t. } \Delta_i < \frac{\epsilon}{4} \ \forall i > i_r(\omega)) = 1 $$
Intuitively, this means almost surely, the realized sequence of payoff matrices has a finite cutoff index after which all estimation errors remain below $\epsilon/4$.

Because of assumption \ref{assump:full_support}, the pure strategy space $\mathcal{X}$ and $\mathcal{T}$ are both compact and $U$ is continuous. Hence, $(\mathcal{X}, \mathcal{T}, U)$ is a two-player zero-sum continuous game, and here are several results that are already proven in \cite{adam2021double} for two-player zero-sum continuous games.
\begin{itemize}
\item Sequences $(\pi_i^*)_{i=1}^{\infty}$ and $(\sigma_i^*)_{i=1}^{\infty}$ have a weakly convergent subsequence, which for simplicity, will be denoted by the same indices. Therefore, $\pi_i^*\Rightarrow \pi^*$ for some $\pi^*$ and $\sigma_i^* \Rightarrow \sigma^*$ for some $\sigma^*$, where $\Rightarrow$ denotes weak convergence.
\item If $\pi_i \Rightarrow \pi$ in $\Delta(\mathcal{X})$ and $\sigma_i \Rightarrow \sigma$ in $\Delta(\mathcal{T})$, then $U(\pi_i, \sigma_i) \to U(\pi, \sigma)$. If $\pi_i \Rightarrow \pi$ in $\Delta(\mathcal{X})$ and $\tau_i \to \tau$ in $\mathcal{T}$, then $U(\pi_i, \tau_i) \to U(\pi,\tau)$.
\item For any $\pi \in \Delta(\mathcal{X})$ we have
        $$\min_{\tau \in \mathcal{T}} U(\pi,\tau) = \min_{\sigma \in \Delta(\mathcal{T})} U(\pi,\sigma)$$
\item The size of initial subset $X_1$ and $Y_1$ can be any finite number.
\end{itemize}
Recall with probability $1$, the realized sequence of payoff matrices has a finite cutoff index after which all estimation errors remain below $\epsilon/4$. From now on, we fix any such realization $\omega$ from this probability $1$ set, and omit the notational dependence on $\omega$ for simplicity. From the proof above, $A_i$ only happens for finite times. Assume $i_r$ is the largest number satisfying that $A_i$ happens. We then treat $(\mathcal{X}_{i_r}, \mathcal{T}_{i_r})$ as the initial set of strategies for both players.  Then our sampling scheme ensures that for any strategy pair $(\pi,\sigma)$ and iteration $i$, we have $|U(\pi_i,\sigma_i) - U(\pi_i,\hat{\sigma}_i)| \leq \Delta_i \leq \epsilon/4$.

Consider any $\mathbf{x}$ such that $\mathbf{x} \in \mathcal{X}_{i_0}$ for some $i_0$. Take an arbitrary $i \geq i_0$, which implies $\mathbf{x} \in \mathcal{X}_i$. Since $(\pi_i^*, \hat{\sigma}_i^*)$ is an equilibrium of the subgame $(\mathcal{X}_i, \hat{\mathcal{T}}_i, U)$, we  get
    $$U(\pi_i^*, \hat{\sigma}_i^*) \geq U(\mathbf{x}, \hat{\sigma}_i^*)$$
    Since $U(\pi_i^*, \sigma_i^*)$ and $U(\pi_i^*, \hat{\sigma}_i^*)$ differ by at most $\frac{\epsilon}{4}$, and $U(\mathbf{x}, \sigma_i^*)$ and $U(\mathbf{x}, \hat{\sigma}_i^*)$ differ by at most $\frac{\epsilon}{4}$, we have
    \begin{align*}
        U(\pi_i^*, \sigma_i^*) + \frac{\epsilon}{2} \geq U(\mathbf{x}, \sigma_i^*) \to U(\mathbf{x}, \sigma^*).
    \end{align*}
    Since $U(\pi_i^*, \sigma_i^*) \to U(\pi^*, \sigma^*)$, we have
    \begin{align}\label{eq:for-closed-x}
        U(\pi^*, \sigma^*) + \frac{\epsilon}{2} \geq U(\mathbf{x}, \sigma^*)
    \end{align}
    for all $\mathbf{x} \in \cup \mathcal{X}_i$. Since $U$ is continuous, the previous inequality holds for all $\mathbf{x} \in cl(\cup \mathcal{X}_i)$.

     Fix now an arbitrary $\mathbf{x} \in \mathcal{X}$. Note $\mathbf{x}_{i+1}$ is the best response to $U(\cdot, \hat{\sigma}^*_i)$ (since ranger oracle uses finite sample estimation of payoff matrix), and we have
    $$U(\mathbf{x}_{i+1}, \hat{\sigma}_i^*) \geq U(\mathbf{x},  \hat{\sigma}_i^*)$$

    Because $U(\mathbf{x}_{i+1}, \sigma_i^*)$ and $U(\mathbf{x}_{i+1}, \hat{\sigma}_i^*)$ differ by at most $\epsilon/4$, and $U(\mathbf{x}, \sigma_i^*)$ and $U(\mathbf{x}, \hat{\sigma}_i^*)$ differ by at most $\epsilon/4$, we have
    \begin{align*}
        U(\mathbf{x}_{i+1}, \sigma_i^*) + \frac{\epsilon}{2} \geq U(\mathbf{x}, \sigma_i^*) \to U(\mathbf{x}, \sigma^*)
    \end{align*}
     Since $\mathbf{x}_{i+1} \in \mathcal{X}_{i+1}$ and by compactness of $\mathcal{X}$, we can select a convergence subsequence $\mathbf{x}_i \to \tilde{\mathbf{x}}$, where $\tilde{\mathbf{x}} \in cl(\cup \mathcal{X}_i)$. This allows us to use \ref{eq:for-closed-x} to obtain 
    \begin{align*}
        U(\mathbf{x}_{i+1}, \sigma_i^*) \to U(\tilde{\mathbf{x}}, \sigma^*) \leq U(\pi^*, \sigma^*) +\frac{\epsilon}{2}.
    \end{align*}
    Therefore, for any $\mathbf{x}\in X$,
    \[
    U(\mathbf{x},\sigma^*)\leq U(\pi^*,\sigma^*)+\epsilon.
    \]
    Similarly, we have for any $\tau \in \mathcal{T}$,
    \[
    U(\pi^*,\tau)\geq U(\pi^*,\sigma^*)-\frac{\epsilon}{2}.
    \]
    The two sides are not symmetrical because the best response for the poacher doesn't use the finite sample approximation of payoff matrix, thus having a smaller error.  We then conclude the proof.
\end{proof}

We then show that adding the terminating condition, for any $\epsilon>0$, algorithm \ref{alg:double-oracle} can terminate in a finite number of iterations almost surely. Also, when it stops, it converges to a $4\epsilon$-equilibrium with high probability. 

\paragraph{Proof of Item 2} 
\begin{proof}
    Consider now an infinite game, from Item 1 in Theorem \ref{thm:convergence_inf_round}, we know that $\bar{v}_i - \underline{v}_i\rightarrow \epsilon$. Also, our sampling scheme ensures that for any strategy pair $(\pi,\sigma)$, almost surely, there exists a finite cutoff index after which all subsequent iteration $i$ satisfy $|U(\pi_i,\sigma_i) - U(\pi_i,\hat{\sigma_i})| \leq \Delta_i \leq \frac{\epsilon}{4}$.
    This indicates that with $\epsilon>0$, the terminating condition will be satisfied after a finite number of iterations almost surely. Assume that the algorithm ends at the $j$-th iteration. This implies
    $$U(\mathbf{x}_{j+1}, \hat{\sigma}^*_j)- U(\pi^*_j, \hat{\tau}_{j+1}) \in (-2\epsilon, 2\epsilon)$$
Then we have
 \begin{align*}
     U(\pi_j,\sigma_j)&\leq U(\pi_j,\hat{\sigma}_j)+ \Delta_j\\
     &\leq U(\mathbf{x}_{j+1},\hat{\sigma}_j)+\Delta_j\\
     &\leq U(\pi_j,\hat{\tau}_{j+1})+\Delta_j+2\epsilon\\
     &\leq U(\pi_j,\tau_{j+1})+2\Delta_j+2\epsilon\\
     &=\arg\min_{\tau} U(\pi_j,\tau)+2\Delta_j+2\epsilon.
 \end{align*}

Here the first and fourth relation follows from the estimation error of $U$. The second one is because $\mathbf{x}_{j+1}$ is the best response for the ranger. The third one is from the terminating condition and the fifth one comes from the best response for the poacher.
Similarly,
 \begin{align*}
     U(\pi_j,\sigma_j)&\geq U(\pi_j,\sigma_{j+1})\\
     &\geq U(\pi_j,\hat{\sigma}_{j+1})-\Delta_j\\
     &\geq U(\mathbf{x}_{j+1},\hat{\sigma}_j)-\Delta_j-2\epsilon\\
     &\geq U(\mathbf{x}^*,\hat{\sigma}_j)-\Delta_j-2\epsilon,\text{ where }\mathbf{x}^*=\arg\max_{\mathbf{x}}U(\mathbf{x},\sigma_j)\\
     &\geq \arg\max_{\mathbf{x}}U(\mathbf{x},\sigma_j)-2\Delta_j-2\epsilon.
\end{align*}
Here the second and fifth relation comes from the estimation error of $U$. The first one is from the best response of the poacher and the fourth one is from the best response of the ranger. The third one follows from the terminating condition.

For $\Delta_j$, we have 
$$P(\Delta_j \geq \epsilon)\leq \frac{1}{16j^{1+\delta}}\leq \mathsf{prob},$$ where the second relation comes from $j>\frac{1}{16\mathsf{prob}}$. Therefore, we show that $(\pi_j,\sigma_j)$ is a $4\epsilon$-equilibrium with probability at least $1-\mathsf{prob}$.
 
\end{proof}

\section{Experimental Details}
\label{sec:exp-details}

\subsection{Datasets}
\label{appendix:data}
\paragraph{Synthetic Dataset} To better reflect real-world conditions, regions are connected based on a predefined topology. We randomly generate 5,100 graphs, each with 30 nodes and 20 edges. The first 4,800 graphs are used for training, the next 200 for validation, and the remaining 100 for testing. Each node is assigned a randomly generated 10-dimensional feature vector. Next, we establish a stochastic mapping from a node’s features to its poaching count, capturing the complex relationships observed in real-world scenarios. Poaching counts are sampled from a Gamma distribution parameterized by shape and scale values.  We randomly initialize two Graph Convolutional Networks (GCNs). For each node, one of the two GCNs is selected with equal probability to map the node’s features to a continuous value, which is then scaled by a factor of 20. This value serves as the shape parameter of the Gamma distribution. The poaching count is then drawn from the Gamma distribution, where the scale parameter is set to 1 if the first GCN is chosen and 0.9 if the second is chosen.  To incorporate adversarial noise, we apply perturbations inversely proportional to the poaching count—nodes with lower poaching counts receive higher noise levels. Finally, the poaching count for each node is capped within the range \([0,40]\) and scaled by 0.2 to align the overall distribution with real-world data.

\paragraph{Real-world Dataset}
We use poaching data from Murchison Falls National Park (MFNP) in Uganda, collected between 2010 and 2021. The protected area is discretized into 1 × 1 km grid cells. For each cell, we measure ranger patrol effort (in kilometers patrolled) as the conditional variable for the diffusion model, while the monthly number of detected illegal activity instances of each cell serves as the adversarial behavior. Following~\citep{basak2016abstraction}, we represent the park as a graph to capture geospatial connectivity among these cells. To focus on high-risk regions, we subsample 20 subgraphs from the entire graph. Specifically, at each time step we identify the 20 cells with the highest poaching counts. Each of these cells is treated as a central node, and we iteratively add the neighboring cell with the highest poaching count until the subgraph reaches 20 nodes. This procedure yields 532 training samples, 62 validation samples, and 31 test samples.

\subsection{Implementation details}
\label{appendix:details}
We use a three-layer Graph Convolutional Network (GCN)~\citep{kipf2022semi} with a hidden dimension of 128 as the backbone of the diffusion model. The diffusion process follows the DDPM framework~\citep{ho2020denoising} with \( T = 1000 \) time steps and a variance schedule from \( 10^{-4} \) to \( 0.02 \). Optimization is performed using Adam~\citep{kingma2014adam} with a learning rate of \( 10^{-3} \), and the model is trained for 5000 epochs. To estimate the expected utility, we draw 500 samples from the diffusion model. All comparison methods run for 30 iterations. The mirror ascent oracle uses a step size of 0.1 and runs for 100 iterations. The step size in the mirror ascent step for the baselines is also 0.1.

The actions of the poacher and ranger in grid $j$, represented by $z_j$ and $x_j$ respectively, influence the wildlife population in the area. We model the wildlife population in grid $j$ as follows: 
$$\max(N_0(j)e^r-\alpha e^{\psi \mathbf{z}_j - \theta \mathbf{x}_j}, 0),$$
where $N_0(j)$ is the initial wildlife population in the area and $r$ denotes the natural growth rate of the wildlife. The parameter $\alpha$ captures the impact of both the ranger’s and poacher’s actions on the wildlife population, $\psi$ reflects the strength of poaching, and $\theta$ measures the effectiveness of patrol effort. The utility for the ranger is then represented as the sum of wildlife population across all grids:
$$U(\mathbf{x}, \mathbf{z}) = \sum_{j=1}^K \max(N_0(j)e^r-\alpha e^{\psi \mathbf{z}_j - \theta \mathbf{x}_j}, 0)$$.

\paragraph{Forecasting Experiments.} 
We use the poaching dataset described in Appendix~\ref{appendix:data}. 
Following~\citet{pmlr-v161-xu21a}, linear regression and Gaussian processes predict the poaching count for each $1\times 1$ km cell individually, using two features: the previous month's patrol effort in the current cell and the aggregated patrol effort from neighboring cells.
For linear regression, we employ the scikit-learn implementation, while for Gaussian processes, we use the GPy library with both the RBF and Matérn kernels. The training procedure for the diffusion model follows Appendix~\ref{appendix:details}, with its support constrained to $[0,3]$. For each test instance, we generate 500 samples and use the mean prediction.  We also attempted to impose constraints on the baseline's output but found that this only degraded its performance.

\end{document}